\documentclass[11pt,oneside,envcountsame,envcountsect]{llncs}
\usepackage{amsmath}
\usepackage[dvips]{pstcol}
\usepackage{amssymb}
\usepackage{graphics}
\usepackage{enumerate}
\usepackage{gastex} 
\usepackage{pict2e}
\usepackage{slashbox}
\usepackage{xspace}
\usepackage[algoruled,vlined,english]{algorithm2e}
\usepackage{algorithmic}
\usepackage[small,hang]{caption2}

\renewcommand{\belowcaptionskip}{-0.5cm}

\begin{document}
\title{On Distributed Computing with Beeps} \author{Y. M\'etivier,
  J.M. Robson and A. Zemmari} \institute{Universit\'e de Bordeaux - 
Bordeaux INP\\
  LaBRI UMR CNRS 5800\\ 351 cours de la Lib\'eration, 33405 Talence,
  France\\ \{metivier, robson, zemmari\}@labri.fr } \date{ }
\pagestyle{plain} \maketitle
\begin{abstract}
We consider networks of processes which interact with beeps. Various beeping 
models are used. The basic one, defined by Cornejo and Kuhn \cite{Cornejo10},
assumes that a process can choose either to beep or to listen; if it listens
it can distinguish between silence or the presence of at least one beep.
The aim of this paper is the study of the resolution of paradigms such as 
collision detection, computation of the degree of a vertex,
colouring, or $2$-hop-colouring  in
the framework of beeping models. For each of these problems we present
Las Vegas or Monte Carlo algorithms and we analyse their complexities
expressed    in terms of the number of slots.

We present also efficient randomised 
emulations  of more powerful beeping models on the basic one. We illustrate
emulation procedures with  an efficient degree computation algorithm
in the basic beeping model; this algorithm was given initially in a 
more powerful model.
\end{abstract}
{\bf keywords:} Beeping model, Collision detection, 
Colouring, $2$-hop-colouring, Degree computation, Emulation.

\section{Introduction}
\subsection{The problem}
Distributed graph algorithms are studied according to standard criteria
that are usually formulated: topological restriction (trees, rings, or
triangulated networks ...), topological knowledge (size, diameter
...), and local knowledge to distinguish nodes (identities, port
numbers).  Another important parameter of these algorithms is the
message size: no limit (local model), 
$O(\log n)$ (congest model, where $n$ is the size of the
graph) or $O(1)$.  For each of these criteria or parameters, we study
in particular the number of steps (rounds) 
necessary to obtain the result.
According to the hypotheses, solutions are deterministic or
randomised.

Typically, if we consider the MIS\footnote{Let $G=(V,E)$ be a graph.
  An independent set of $G$ is a subset $I$ of $V$ such that no two
  members of $I$ are adjacent. An independent set $I$ is maximal,
  denoted MIS, if any vertex of $G$ is in $I$ or adjacent to a vertex
  of $I.$} problem, when no identifiers are available there are only
randomised solutions.  Since the major contribution due to Luby
\cite{Luby}, this problem has been extensively studied with parameters
given above.  
 More recently, Afek et al. \cite{Afek13}, inspired by biological
observations, study the MIS problem through the beeping model: at each
step a vertex can either beep (emit a signal) or be silent, and if it is silent
it can distinguish between silence or the presence of at least one beep in its 
neighbouring.  This approach has been developed in several
papers \cite{Cornejo10,Schneider10,Afek13,HuangM13,Scott13} for
distributed problems such as MIS computation, (interval) colouring, conflict
resolution, membership problem etc.  

Let $G$ be a graph and let $v$ be a vertex of $G$; two kinds of collisions
may happen from the point of view of $v$:
\begin{itemize}
\item $v$ beeps and simultaneously at least one 
 neighbour of $v$ beeps, this collision
is called an internal collision;
\item at least 
two distinct neighbours of $v$ beep simultaneously, this collision 
is called a peripheral collision.
\end{itemize}

In this paper, we consider several variants of beeping models:
\begin{itemize}
\item if a process beeps, there are two cases:
\begin{enumerate}
\item it cannot know whether another process beeps simultaneously
 (see \cite{Cornejo10}), 
this case  is denoted
by $B$;
\item it can distinguish whether it beeped alone or if at least one neighbour
beeped concurrently, it is an internal collision; 
this case is called sender side collision detection  in
  \cite{Afek13} Section 6, and it is denoted in this paper $B_{cd}$;
\end{enumerate}
\item if a process listens, there are also two cases:
\begin{enumerate}
\item
it can distinguish between silence or the presence of at least one beep 
(see \cite{Cornejo10}), this
model is denoted $L$;
\item
  it can distinguish between silence or the presence of one beep or
   the presence of at least two beeps; in this case it is a peripheral
collision,
(see \cite{Schneider10},\cite{Afek13} Section 4), this model
is denoted $L_{cd}$  in this paper.
\end{enumerate}
\end{itemize}

Finally, a beeping model is defined by choosing between $B$ or $B_{cd}$ and
between $L$ and $L_{cd}$. For example the basic beeping model introduced
by Cornejo and Kuhn in \cite{Cornejo10} is $BL$; Afek et al. in \cite{Afek13}
(Section 6)
 and Scott et al. in \cite{Scott13} study the MIS  problem in the
model $B_{cd}L$. In Section 4 of \cite{Afek13}, Afek et al. study the MIS 
problem in $BL_{cd}$.
In this paper we present algorithms in models $BL$, $B_{cd}L$ and $B_{cd}L_{cd}$.

  Usually, the topology of a distributed system is modelled
by a graph and para\-d\-i\-g\-ms of distributed systems are represented by
classical problems in graph theory
such as  vertex degree,  maximal independent set (MIS for short),  
$2$-MIS  (we recall that a $2$-MIS of a graph $G$
is a MIS of the square of $G$, i.e., the graph with the set of vertices
of $G$ in which there is an edge between any two different
vertices $u$ and $v$ if the distance
between $u$ and $v$ in $G$ is at most $2$), colouring
(a colouring  of a graph $G$ assigns colours to vertices
such that two neighbours have different colours),
$2$-hop-colouring
(as for  a $2$-MIS,
a $2$-hop-colouring of a graph $G$ is a colouring of the square of $G$).
 Each solution to one of these problems is a
building block for many distributed algorithms: symmetry breaking,
topology control, routing, resource allocation or network
synchronisation. 

As explained in \cite{Peleg} (p. 79), a MIS or a colouring enables 
the construction of schedules such that two neighbouring vertices do not act
concurrently. Furthermore, a MIS can help for the decomposition of a network
into clusters. A $2$-MIS 
makes it possible
 to assign each vertex to exactly one leader.
Channel assignment for a radio network with collision-freedom corresponds
to a $2$-hop-colouring of the graph corresponding to the network  since
each colour corresponds to a channel \cite{KMR01}.
The importance of the $2$-hop-colouring is also attested by 
Emek et al. \cite{EPSW14}, they 
 prove that in an anonymous network 
any randomised algorithm can be seen as
the composition of a randomised $2$-hop-colouring and a deterministic
algorithm.
Finally, in an anonymous 
wireless network there are no port numbers, in this context a $2$-hop-colouring 
ensures that no node has two neighbours with the same colour, and colours act
as port numberings.

The aim of this work is the study of the resolution of 
these problems in the framework
of beeping models.

In this paper,  results on graphs
having $n$ vertices are expressed with high probability (w.h.p. for short),
meaning with probability $1-o(n^{-1})$.

Let $G$ be a graph and let $v$ be a vertex of $G$. 
We denote by $\Delta$ the maximum degree of $G$.
The neighbourhood of
$v$, denoted $N(v)$, is the set of vertices adjacent to $v$ (at distance
$1$ from $v$).
We define $\overline{N}(v)$ 
 by including $v$ itself in $N(v)$.
We use also the set of vertices at distance at most $2$ from $v$ called
the $2$-neighbourhood and denoted $N_2(v)$.
We write $\log n$ for the natural logarithm of $n$ and $\log_2 n$ for
the logarithm of $n$ to the base $2$.

\subsection{The Network Model}
We consider a wireless network model and
we follow definitions given in  \cite{Cornejo10} and in \cite{Afek13}.
The network is anonymous: unique identities are not available to
distinguish the processes. 
The network communications are
synchronous and encoded by a connected graph $G=(V,E)$ where the vertices
$V$ represent processes and the edges $E$ represent pairs of processes
that can hear each other.
We assume that all processes wake up and start computation at the same
step.
Time is divided into discrete synchronised time intervals, and
during each time interval all processors act in parallel and:
\begin{itemize}
\item beep or listen;
\item perform local computations.
\end{itemize}
Usually, in the message passing point to point model each interval is called
a round, and in the context of wireless network model each interval is
called a slot.
\begin{remark}\label{listen}
In general, vertices are active or passive. When they are active
they beep or listen; in the description of algorithms we say explicitely when 
a vertex beeps meaning that a non beeping active vertex listens.
\end{remark}

The time complexity, also called the slot complexity,
 is the maximum number of slots
needed until every vertex has completed its computation.

Algorithms are expressed with a for-loop or an until-loop;
in this paper, we call a phase one  execution of the body of the for-loop 
or of the until-loop.

\begin{remark}
An algorithm given in the beeping model induces an algorithm in the
message passing model; thus any lower bound on the round complexity
in the message passing model is a lower bound on the number of
slots in the beeping model.
\end{remark}

\subsection{Distributed Probabilistic Algorithm}
A probabilistic algorithm is an algorithm which makes some random choices
 based on some given probability distributions. 

A distributed probabilistic algorithm is a collection of local probabilistic
algorithms. The network is anonymous, and 
processes have no information on their  degrees;
thus their local probabilistic algorithms
are identical and have the same probability distribution.

A Las Vegas algorithm  is a probabilistic algorithm which terminates
with a positive probability (in general $1$) and always produces
 a correct result.

A Monte Carlo algorithm is a probabilistic algorithm which always
terminates; nevertheless the result may be incorrect with a certain
probability.

\subsection{Our Contribution}
Classical considerations on symmetry breaking in anonymous 
beeping networks, see for example
\cite{Afek13} (Lemma 4.1) , imply that:
\begin{remark}
There is no Las Vegas internal collision
detection algorithm in the beeping models $BL$ and $BL_{cd}$.
There is no Las Vegas peripheral  collision
detection algorithm in the beeping models $BL$ and $B_{cd}L$.
\end{remark}

Finally, a first contribution may be summarised by the following table.
\vskip 0.3cm

\begin{minipage}{\linewidth}
\begin{center}
\begin{tabular}{|l*{4}{|c}|}
\hline
\backslashbox{\parbox[l]{1.5cm}{Problem}}{\parbox[l]{1.5cm}{Model}}
 & $BL$ & $B_{cd}L$ & $BL_{cd}$ & $B_{cd}L_{cd}$\\
\hline
Collision Detection & MC & MC&MC  &LV \\
\hline
Degree & MC &MC &MC & LV \\
\hline
Colouring & MC & LV & MC & LV \\
\hline
2-colouring & MC & MC & MC & LV \\
\hline
\end{tabular}

\vspace{0.5cm}
MC means there exists a Monte Carlo algorithm and there exists no 
Las Vegas algorithm.

 LV means there exists a Las Vegas algorithm.
\bigskip
\end{center}
\end{minipage}

\subsubsection{Collision Detection.}
We present and analyse  very simple Monte Carlo 
procedures which detect internal and peripheral 
collisions in the beeping model
 $BL$.

Let  $G$ be a graph and let $v$ be a vertex of $G$.
According to the initial knowledge (error probability $\epsilon$
and/or 
the size of the graph),
we prove that, given $0<\epsilon <1$,
any collision in   $N(v)$ 
is detected in  
$O\left(\log(\frac 1 \epsilon)\right)$ slots with an error probability upper
 bounded by  $\epsilon$ or in $O(\log n)$ slots with an error probability
 $1-o(\frac 1 {n^2})$.
Any collision in $G$ is detected
in $O(\log (\frac n \epsilon))$ slots with an error probability upper bounded
by $\epsilon$ 
and in $O(\log n)$ slots with probability $1-o\left(\frac 1 {n}\right)$,
i.e., w.h.p.

\subsubsection{Colouring and 
$2$-hop-colouring Algorithms.}
Algorithms for   colouring and $2$-hop-colouring are based
on a repeat-loop whose body has three parts:
\begin{enumerate}
\item a vertex is candidate  to a colour
and beeps 
with a certain probability which can change after
each iteration,
\item a candidate vertex tries to detect whether it is the only
candidate  or not
in $\overline{N}(v)$  or $N_2(v)$,
\item according to the conclusion, it informs its neighbours (and
possibly neighbours of its neighbours) and
may adjust its probability to be once again candidate.
\end{enumerate}

We present and analyse a Las Vegas colouring
algorithm in the model $B_{cd}L$; its slot complexity is 
$76 \log_2 n + 112\Delta$. We present also a  $2$-hop-colouring Las Vegas
algorithm in the model $B_{cd}L_{cd}$; its slot complexity
is  $5\times(76 \log_2 n + 112\Delta^2)$.
  In both cases algorithms need
no knowledge on $G$.

In the case where we know an upper bound $K$ on the maximum degree 
of the graph we provide a colouring algorithm with colours
bounded by $K+1$ and with a slot complexity equals to 
$O\left( K(\log n +\log^2K)\right)$.
\subsubsection{Emulation.}
Based on results of the section devoted to collision detection we
propose emulation procedures of $B_{cd}$ and of $L_{cd}$ in $BL$. 
Let $G$ be a graph.
Each beep or listen is emulated by 
 $k=2\times\lceil \log_2\left(\frac n\varepsilon\right) \rceil$ slots, 
and 
 the procedures are correct on $G$ with
probability $1-\varepsilon$, 
or by  $k=2\times\lceil \log_2\left(\frac 1 \varepsilon\right) \rceil$
slots, and,
 for any vertex $v$, the procedures are correct on $v$ 
with probability $1-\varepsilon$, or by
 $k = 2\times\lceil 2\log_2 n\rceil$ slots,  and 
 the procedures are correct on $G$ w.h.p. 
Finally, emulation procedures induce a logarithmic multiplicative factor for the
slot complexity.

\subsubsection{Degree Computation.}
First,
we deduce from the $2$-hop-colouring a Las Vegas degree computation algorithm
in $B_{cd}L_{cd}$; its slot complexity is  $5\times(76 \log_2 n + 112\Delta^2)$.

We illustrate emulation procedures by applying them
to the degree computation algorithm given in $B_{cd}L_{cd}$ and
we obtain a Monte Carlo algorithm for
 the computation of the degrees of each vertex in $BL$.
For any graph $G$ of size $n$,
the new algorithm computes the degrees in $G$ in 
$O\left((\log n+\Delta^2)\log n\right)$, and the result is correct 
w.h.p.

\begin{remark}
For some problems, the design of some algorithms is more natural and easier
in $B_{cd}L_{cd}$ than in $B_{cd}L$ or  is more natural and easier
in $B_{cd}L$ than in $BL$. In these cases emulation procedures enable safe and
automatic translations of algorithms given in a strong model into a weaker
model.
\end{remark}

\begin{minipage}{\linewidth}
\begin{tabular}{|p{2.6cm}|p{2.cm}|p{4.cm}|p{4.2cm}|p{2.2cm}|}
\hline

Problem &Beeping model &Time (number of slots) & Information required at each node &  error probability \\

\hline

Collision detection in $N(v)$ &$BL$ & $O\left(\log(\frac 1 \epsilon)\right)$ &  $\epsilon$ & Monte Carlo  at most  $\epsilon$ \\
\hline

Collision detection in $N(v)$  &$BL$ & $O(\log n)$ &  
size of the graph &Monte Carlo  $o\left(\frac 1 {n^2}\right)$ \\
\hline

Collision detection in $G$  &$BL$ & $O\left(\log(\frac
n \epsilon)\right)$ &  
size of the graph and $\epsilon$ & Monte Carlo  at most $\epsilon$ \\
\hline

Collision detection in $G$  &$BL$ & $O\left(\log n\right)$   &  
size of the graph & Monte Carlo  $o\left(\frac 1 {n}\right)$ \\
\hline
MIS  \cite{Scott13} & $B_{cd}L$  &  $O(\log n)$ & none& 
Las Vegas \\
\hline
Colouring \cite{Cornejo10} & $BL$    & never stops stabilisation 
w.h.p. in $O(\Delta \log n)$ & Each node knows its degree
and an upper bound of  $\Delta$ & Monte Carlo \\
\hline
Colouring  &  $B_{cd}L$  & $O(\log n+\Delta)$ 
w.h.p. & none & Las Vegas \\
\hline
Colouring  &  $B_{cd}L$  & $O\left( K(\log n +\log^2K)\right)$ 
w.h.p. & An upper bound $K$ on the maximum degree  of $G$ & Las Vegas \\
\hline
$2$-colouring  &  $B_{cd}L_{cd}$  & $O(\log n +\Delta^2)$ 
w.h.p. & none & Las Vegas \\
\hline
Degree computation  &  $B_{cd}L_{cd}$  & $O(\log n +\Delta^2)$ 
w.h.p. & none & Las Vegas \\
\hline
Degree computation  &  $BL$  &$O\left((\log n + \Delta^2)(\log(\frac{n}{\varepsilon}))\right)$ 
 & size of the graph and $\varepsilon$ & Monte Carlo at most $\varepsilon$\\
\hline
Degree computation  &  $BL$  &$O\left((\log n+\Delta^2)\log n\right)$
 & size of the graph & Monte Carlo $o(\left( \frac 1 n \right))$\\
\hline
\end{tabular}\par
\vspace{0.5cm}
Beeping algorithms on graphs with $n$ vertices.
\bigskip

\end{minipage}

\subsection{Related Work}

As explained by Chlebus \cite{C01}, in a radio network, a vertex can hear
a message only if it was sent by a neighbour and this neighbour was the only
neighbour that performed a send operation in that step. If no message
has been sent to a vertex then it hears the background noise. If a vertex
$v$ receives more than one message then we say that a collision occurred at the
vertex $v$ and the vertex hears the interference noise. If vertices
of a network can distinguish the background noise from the interference noise
then the network is said to be with collision detection, otherwise it is
without collision detection (see for example  the Wake-up problem or 
the MIS problem for radio networks 
in \cite{GPP01,Moswa,CGK07,JK15} where vertices do not 
make the difference  between no neighbour sends a message and at least
two neighbours send a message; see also the broadcasting problem
in radio network in \cite{GHK13} where vertices make the difference between
no neighbour sends a message, exactly one neighbour send a message and 
at least two neighbours send a message).
In this context, an efficient randomised
emulation of single-hop radio network with collision detection on
multi-hop radio network without collision detection is presented and
analysed in \cite{BGI91}.
To summarise:
\begin{remark}
Detecting  a collision in a radio network is to be able to distinguish
between $0$ message and at least $2$ messages while detecting a collision
in the beeping model is to be able to distinguish between $1$ message
and at least $2$ messages.
\end{remark}

Thus, from now on, we consider collisions as explained above for
beeping models.  Our collision detetection algorithm 
and  the degree computation algorithm use
similar ideas to those used for initialising a packet radio network
\cite{HNO99} or for election in a complete graph with wireless
communications \cite{BW12} (Algorithm 50, p. 132).  The impact of
collision detection is studied in \cite{Schneider10,KP13}, where it is
proved that performances are improved, and in certain cases the
improvement can be exponential.  The complexity of the conflict
resolution problem (the goal is to let every active vertex use the
channel alone (without collision) at least once) is studied in
\cite{HuangM13} (they assume that vertices are identified), and an
efficient deterministic solution is presented and analysed.

General considerations and many examples of Las Vegas distributed
algorithms related to MIS or colouring can be found in \cite{Peleg}.
The computation of a MIS has been the object of extensive research on
parallel and distributed complexity in the point to point message
passing model \cite{Alonco,Luby} \cite{Awerbuchco,Linial}; Karp and
Wigderson \cite{KarpW} proved that the MIS problem is in NC.
Some links with distributed graph colouring and some recent results on
this problem can be found in \cite{Kuhnco}.  The complexity of some
special classes of graphs such as growth-bounded graphs is studied in
\cite{Kuhnmo}.  Results have been obtained also for radio networks
\cite{Moswa}.  A major contribution is due to Luby \cite{Luby}.  He
gives a Las Vegas distributed algorithm.  The main idea is to obtain
for each vertex a {\it local total order} or a {local election} which
breaks the local symmetry and then each vertex can decide locally
whether it joins the MIS or not.  Its time complexity is $O(\log n)$
and its bit complexity is $O(\log^2 n).$ Recently, a Las Vegas
distributed algorithm has been presented in \cite{MRSZ11} which
improved the bit complexity: its bit complexity is optimal and equal
to $O(\log n)$ w.h.p.  
An experimental comparison between 
\cite{Luby} and \cite{MRSZ11} is presented in \cite{BK13}.
If we remove the constraint on the size of
messages or on the anonymity recent new results have been obtained for
distributed symmetry breaking (MIS or colouring) in
\cite{KothaP11,BarenEPS12,BarenE13,BarenE14}.

Afek et al. \cite{Afek13}, 
 from considerations concerning the development of
certain cells, studied the MIS problem in the discrete beeping model 
 $BL$
as presented in \cite{Cornejo10}. They consider, in particular,
the wake-on-beep model (sleeping nodes wake up upon receiving a beep) and
sender-side collision detection $B_{cd}L$: they give an $O((\log n)^2)$ rounds 
MIS algorithm. After this work, 
Scott et al. \cite{Scott13} presents in the model $B_{cd}L$ 
a randomised algorithm
with feedback mechanism  whose expected time to compute a MIS is
$O(\log n)$.  A vertex $v$ is candidate for joining the independent set
(and beeps) with a certain probability 
(initially $1/2$); this value is decreased by some fixed factor if at
least one neighbour whishes also to join the independent set. It is increased
by the same factor (up to maximum $1/2$) if neither $v$ nor
any  neighbour of $v$ are candidates.

More generally, Navlakha and Bar-Joseph present in \cite{NB15}
a general survey on 
similarities and differences between distributed computations in biological
and computational systems and, in this framework, the importance
of the beeping model.

In the model of point to point 
message passing,
vertex colouring is mainly studied under two assumptions: - vertices
have unique identifiers, and more generally, they have an initial
colouring, - every vertex has the same initial state and initially
only knows its own edges.  If vertices have an initial colour, Kuhn
and Wattenhofer \cite{Kuhnco} have obtained efficient time complexity
algorithms to obtain $O(\Delta)$ colours in the case where every
vertex can only send its own current colour to all its neighbours.  In
\cite{Johansson}, Johansson analyses a simple randomised distributed
vertex colouring algorithm for anonymous graphs.  He proves that this
algorithm runs in $O(\log n)$ rounds w.h.p. on graphs
of size $n.$ The size of each message is $\log n,$ thus the bit
complexity per channel of this algorithm is $O(\log^2 n).$
\cite{MRSZ10} presents an optimal bit and time complexity Las Vegas
distributed algorithm for colouring any anonymous graph in $O(\log n)$
bit rounds w.h.p.

In \cite{Cornejo10}, Cornejo and Kuhn study the interval colouring problem:
an interval colouring assigns to each vertex an interval (contiguous fraction)
of resources such that neighbouring vertices do not share resources
(it is a variant of vertex colouring). They assume that each node
knows its degree and an upper bound of the maximum degree $\Delta$
of the graph. They present in the  beeping
model $BL$ a probabilistic  algorithm which never stops
and stabilises with a correct  $O(\Delta)$-interval coloring  in
$O(\log n)$ periods  w.h.p., where:  $n$ is 
the size of the graph, and a period is $Q$ time
slots with $Q\geq \Delta$, thus it stabilises in $O(Q(\log n))$ slots.

Kothapalli et al. consider the family of anonymous rings and show in
\cite{KOSS} that if only one bit can be sent along each edge in a
round (point to point message passing model), 
then every Las Vegas distributed vertex colouring algorithm (in
which every node has the same initial state and initially only knows
its own edges) needs $\Omega(\log n)$ rounds w.h.p. to
colour the ring of size $n$ with any finite number of colours.
Kothapalli et al.  consider also the family of oriented rings and they
prove that the bit complexity in this family is $\Omega(\sqrt{\log
  n})$ w.h.p.

\cite{FMRZ13} presents and analyses Las Vegas distributed algorithms which compute
a MIS or a maximal matching for anonymous rings (in the point to 
point message passing model).  Their bit complexity
and time complexity are $O(\sqrt{\log n})$ w.h.p.

Emek and Wattenhofer introduce in \cite{EmekW13}
a model for distributed computations which resembles the beeping model:
networked finite state machines (nFSM for short). This model enables the sending
of the same message to all neighbours of a vertex; however it is asynchronous,
the states of vertices belong to a finite set, the degree of vertices is bounded
and the set of messages is also finite. In the nFSM model they give a
$2$-MIS algorithm for graphs of size $n$ using a set of messages of size $3$
with a time complexity equal to $O({\log n}^2).$

\section{A Monte Carlo Collision Detection Algorithm 
in $BL$}\label{collision}

If we consider the  beeping models presented in the Introduction, clearly
the weakest is $BL$.
This section presents simple and efficient
probabilistic procedures for detecting collisions
by using  $BL$. Later (Section \ref{emuler}) we will see how 
 to emulate  $B_{cd}$
or  $L_{cd}$ in $BL$.

A phase $P$ is  the sequence of the 3 following actions:
\begin{itemize}
\item vertices wishing to beep, randomly and uniformly select 0 or 1;
\item slot $1$: vertices that have drawn 0 beep, the others listen;
\item slot $2$: vertices that have drawn  1 beep, the others listen.
\end{itemize}

A vertex detects a collision if:
\begin{itemize}
\item it does not beep and it hears beeps at two slots 
in a phase,
\item or if it beeps itself at a slot of a phase and hears a beep
at the other slot of the same phase.
\end{itemize}
We address two questions:

Let $0<\epsilon<1$, how many phases must each vertex execute to decide
whether  there is a collision or not in its neighbourhood 
with an error probability bounded by
$\epsilon$?

Let $0<\epsilon<1$, how many phases must each vertex execute to ensure
that whether  there is a collision or not over all the graph $G$
is detected with an error
probability bounded by $\epsilon$?


\begin{algorithm}[h]
\caption{Collision Detection Algorithm in $BL$ - according to the desired
 error probability
and the knowledge of vertices,
$k=\lceil {\log_2(\frac 1 \epsilon)}\rceil +1 $ or 
$k=\lceil {2\log_2(n)}\rceil +1$ or
$k=\lceil {\log_2(\frac n \epsilon)} \rceil +1$.}
\textbf{Var:}\\
\hspace*{1cm}$k:$ \textbf{Global integer constant;}\\
\hspace*{1cm}$collision:$ $boolean$ \textbf{Init } $false$;\\
\hspace*{1cm}$i:$ $Integer;$\\
\hspace*{1cm}$b:$ \textbf{in} $\{0,1\}$;\\
\For{$i:=1,k$}{
\If{$v$ wishes to beep}
   {Choose $b$ uniformly at random from $\{0,1\}$; \\
   \If{$b=0$}
      {slot 1 beep; slot 2 listen}
      \Else{slot 1 listen; slot 2 beep;}
    \If{a beep was heard}{$collision:=true$}
   }
   \Else{slot 1 listen; slot 2 listen;\\
   \If{two beeps were heard}{$collision:=true;$}
   }
}
\end{algorithm}

We have:
\begin{lemma}
\label{lemma::collision_local}
Let $G$ be a graph having $n$ vertices.
Let $v$ be any vertex. Let $0<\epsilon<1$.
 Any collision in the neighbourhood of $v$ is
detected in $O\left(\log_2(\frac 1 \epsilon)\right)$ phases (slots)
with
probability at least $1-\epsilon,$ and in $O\left(\log_2 n\right)$ phases
(slots) with probability $1-o\left(\frac 1 {n^2}\right)$.
\end{lemma}
\begin{proof}
Let $v$ be any vertex having $d(v)\geq 1$ neighbours. If a collision
happens between $u_1$, which is either $v$ or a neighbour of $v$ and $u_2$, 
a neighbour of $v$, then it will be detected if and only if $u_1$  chooses 
a slot different from $u_2$. This happens with probability $1/2$.

Thus, the probability that a collision happens and is not detected in
the neighbourhood of $v$ within next $k$ phases is at most
$\left(\frac 1 {2}\right)^k$. This probability is then less than
$\epsilon$ (resp. less than $o\left(\frac 1 {n^2}\right)$) for any
$k>\log_2(\frac 1 \epsilon)$
(resp.  $k>2\log_2(n)$), which ends the proof.
\qed
\end{proof}
Yielding:
\begin{corollary}
Let $G$ be a graph having $n$ vertices.
Any collision in $G$ is detected after at most $O\left(\log_2(\frac
n \epsilon)\right)$ phases (slots) with probability at least $1-\epsilon$, and
after at most $O\left(\log_2 n\right)$ phases (slots) with probability
$1-o\left(\frac 1 {n}\right)$.
\end{corollary}
\begin{proof}
Assume a collision occurs at time $t_0$ in $G$ and let $T$ denote the number
of phases before it is detected in the whole graph. Clearly $T=\max\{
T_v\mid v\in V\}$, where $T_v$ denotes the time before a node $v$
detects a collision in its neighbourhood and then:

\begin{eqnarray}
{\mathbb P}r
\left(
T>\log_2\left(\frac n \epsilon\right)\right)
& \leq & n\times 
{\mathbb P}r
\left(
T_v>\log_2\left(\frac n \epsilon\right)\right)\\
& = & n\times \frac 1 {2^{\log_2(\frac n \epsilon)}}= \epsilon.
\end{eqnarray}
Which proves the first claim. The same argument, combined with the
second claim of Lemma \ref{lemma::collision_local} proves the second
claim of the corollary.  \qed
\end{proof}
These results can be  summarised by  Algorithm 1 (Monte Carlo).

\section{Colouring Algorithms}
\subsection{A Las Vegas 
Colouring Algorithm in $B_{cd}L$ without 
any knowledge}\label{colouring}

This section presents and analyse a Las Vegas colouring algorithm
in the model $B_{cd}L$
assuming that the vertices have no knowledge.

Initially each vertex is active.
Each active vertex $v$ maintains a parameter $p$, its ``beeping
probability'' initially equal to $1/2$.  
It maintains also a counter, denoted colour (initially equal to $0$),
 that is incremented at each iteration.
In each phase each active
vertex decides with probability $p$ to beep, indicating that it
is a candidate to the current colour given by the counter.
It succeeds and its colour is the value of the counter if and only if no
neighbour has also beeped; in this case its state becomes coloured.
Then after this slot,
 if $v$ is still active, it adjusts
$p$, halving it if any neighbour beeped and doubling it if no
neighbour beeped and it is not already $1/2$.
If a neighbour has
beeped we say that $v$ is ``inhibited''.

\begin{algorithm}\label{Scott}
\caption{A Las Vegas colouring algorithm without any knowledge in $B_{cd}L$.}
\textbf{Var:}\\
\hspace*{1cm}$state\in\{active,coloured\}$ \textbf{Init } $active$; \\
\hspace*{1cm}$candidate:$ $Boolean$;\\
\hspace*{1cm}$p:$ $real$ \textbf{Init } $1/2$;\\
\hspace*{1cm}$colour:$ $Integer$ \textbf{Init } $0$;\\
\Repeat{$state=coloured$}{
  $colour:=colour+1$;\\
        set $candidate$ to $true$ with probability $p$ else $false;$\\
        \If{$candidate$}{beep}
        \If{$candidate$ and no internal collision}{
          $state:=coloured$}
         \If{$state=active$}{
            \If{$not$ $candidate$ and no beep heard}{
            \If{$p<1/2$}{$p:=2\times p$}}
            \Else{$p:=p/2$}
}
}
\end{algorithm}
\begin{remark}
At the end of the body of the until-loop, we can add a slot which 
enables an uncoloured vertex  to beep and finally a couloured vertex
can detect the local termination of the colouring algorithm.
\end{remark}

We first introduce some notation that we will use in this proof.
For any vertex $v$, $p_v$ denote the parameter $p$ on the vertex 
$v$ and we define the following sum:
$$ q_v = \sum_{u\in N(v)}p_u.
$$ 

We also note $q_v^* = \max\{q_v,1/5\}$ and finally $t_0 = 3 \log_2
(5q_v^*)- 2 \log_2 p_v$. We omit the subscript $v$ where there is no
risk of ambiguity.

We finally write $l(q)$ for $\log_2 (5\max\{q,1/5\})$, that is
$l(q)=\max\{\log_2(5q),0\}$.

Recall that $\overline{N}(v)$ is the  set of vertices at distance
less than or equal to $1$ from vertex $v$.

Then, we have the following theorem:
\begin{theorem}
For any $t\ge 0$ and for any vertex $v$, its probability of remaining
active after the next $t$ phases is at most $\alpha^{112d(v)+t_0-t}$ for the
constant $\alpha=2^{1/36} \approx 1.01944$, where $d(v)$ is the degree of $v$
in the residual graph.
\end{theorem}
Note that $\alpha^{3\log_2q}=q^{3\log_2\alpha}=q^{1/12}$.  The proof
will be by induction on $t$. We have $t_0\ge 2$, so that if $t=0$,
$\alpha^{t_0-t} >1$ and the claim is trivially true.

Let $t>0$. After one phase which does not colour $v$
we have by induction that the probability of remaining active
for the following $t-1$ phases is at most $\alpha^{112d^\prime(v)+t^\prime_0-t+1}$
where $t^\prime_0$ is the new value of $t_0$, namely $ 3
l(q^{\prime})- 2\log_2 p^\prime$ and $d^\prime(v)$ is the new degree.
So we conclude that the probability
of survival is upper bounded by the mean of the random variable which
is $\alpha^{112d^\prime(v)+t^\prime_0-t+1}$ if $v$ survives the first phase and $0$
otherwise. We refer to this mean as the {\it bound} and note that it
is dependent on what happens outside the neighbourhood of $v$.

We will come back to the proof of the Theorem, but we first prove the
following lemma:
\begin{lemma}
The bound is maximised when what happens outside the neighbourhood of
$v$ is that every neighbour $u$ of $v$ is inhibited from taking the
current colour by an external neighbour beeping.
\end{lemma}
Proof\\

Consider any external behaviour $E$ in which some $u$ is not
inhibited; we will show that the bound is increased or unchanged if
the behaviour is changed to $E^\prime$ in which $u$ is inhibited and
there is no change for any other neighbours of $v$.  (In a given graph
there may be no such $E^\prime$ but we consider the maximum possible
over any graph containing the neighbourhood $\overline{N}(v)$.)  We
consider fixed beeping decisions of all vertices in $\overline{N}(v)$
except $u$ and show that with these decisions $E^\prime$ gives a value
of the bound greater than or equal to that of $E$.  We consider two
cases:

\begin{itemize}
\item Some neighbour of $u$ in $\overline{N}(v)$ beeps:\\ $p_u$ will
  be halved whether or not $u$ is inhibited by $E^{\prime}$ and so
  $p^\prime$, $q^\prime$, $d^\prime(v)$ and the probability of survival are
  the same for $E$ and $E^\prime$. The bound is identical in the two
  cases.
\item Otherwise: 

Let the value of $p^\prime$ be $p_0$ if $u$ does not beep and $p_1$ if
$u$ does beep. $p_1 \le p_0$.

Let the value of $q^\prime$ be $q_0$ if $u$ does not beep and is not
inhibited, $q_1$ if it beeps and is inhibited and $q_2$ if it does not
beep and is inhibited.
Note that if $u$ beeps and is not inhibited,
$u$ takes the current colour; we note the value of $q^\prime$ in this case
as $q_3$ and note that $q_3<q_0$ since the effect of $u$ beeping is to
remove $p_u$ from the sum for $q$ and possibly to halve the values of $p$ for some
common neighbours of $u$ and $v$.
We have
$q_1\ge q_0/4$ since, at most, $u$'s beeping can result in a vertex
$w$ halving $q_w$ when otherwise it would have doubled it. Similarly
$q_2\ge q_0/4$ and $q_2 \ge q_0-3p_u/2$ since the inhibition results in
$p_u$ being halved rather than potentially doubled.

Let $d_0$ be the new value of $d(v)$ if $u$ does not take the current colour;
if it does, then the new value is $d_0-1$.

The bounds are thus $p_u\alpha ^ {112d_0+ 3 l(q_1)- 2 \log_2(p_1)-t+1}+
(1-p_u)\alpha ^ {112d_0+ 3 l(q_2)- 2 \log_2(p_0)-t+1}$ in the inhibited case
and
$(1-p_u)\alpha ^ {112d_0+ 3 l(q_0)- 2 \log_2(p_0)-t+1}+p_u\alpha^{112(d_0-1)+3l(q_3)-2\log_2(p_1)}$
in the
uninhibited case.  We claim that the ratio of the inhibited bound to
the uninhibited is at least $1$.  This ratio $\ge \frac{p_u\alpha^{ 3
    l(q_1)}+(1-p_u)\alpha^{ 3 l(q_2)}} {(1-p_u)\alpha^{ 3 l(q_0)}+p_u\alpha^{3l(q_0)}/8}$
(since $p_1\le p_0$, $p_1\ge p_0/4$, $q_3<q_0$) and $\alpha^{108}=8$\\ 
Remember that $p_u$ is a power of $1/2$.  We
consider four subcases:
\begin{itemize}
\item $q_0 \le 1/5$: $l(q_1)=l(q_2)=l(q_0)=0$ and the ratio
$\ge (p_u+1-p_u)/(1-p_u+p_u/8) > 1$.

\item $1/5 < q_0$ and $p_u\ge 1/8$: We use the bounds $q_1 \ge q_0/4$ and
$q_2 \ge q_0/4$ giving that the
ratio is at least
$(p_u+1-p_u)\alpha^{-6}/(1-p_u+p_u/8)$
$= \alpha^{-6}/(1-p_u+p_u/8) \ge \alpha^{-6}/(7/8+1/64) \ge 1$.

\item $1/5 <q_0 \le 4/5$ and $p_u \le 1/16$: We use the bounds $q_1
  \ge q_0/4$ and $q_2 \ge q_0 -3p_u/2$ and the fact that for $0<x \le
  15/32$, $(1-x)^{1/12} > 1-4/3(x/12)$ so that the ratio is at least
  $(p_u\alpha^{-6}/(1-p_u) + (1-3p_u/2q_0)^{3\log_2\alpha})\frac{1-p_u}{1-p_u(1-1/8)}$
  $\ge (p_u\alpha^{-6} + (1-15p_u/2)^{1/12})\frac{1-1/16}{1-(1-1/8)/16}$ 
  $\ge (p_u\alpha^{-6} + (1-(15p_u/2)/12\times (4/3)))\frac{120}{121}$ 
  $\ge (1 + p_u(\alpha^{-6} - 5/6) )\frac{120}{121}> 1$.

\item $q_0 > 4/5$ and $p_u\le 1/16$: Using the same bounds as in the
  previous subcase the ratio is greater than
  $(\frac{p_u}{1-p_u}\alpha^{ -6} + \alpha^{ 3 (l(q_0-3p_u/2)-l(q_0))})\frac{120}{121}$
  $>(\frac{p_u}{1-p_u}\alpha^{-6 } + \alpha^{ 3
    (l(4/5-3p_u/2)-l(4/5))})\frac{120}{121}$ and this is the bound already used for
  the case with $q_0=4/5$ and the same value of $p_u$ and so is
  greater than or equal to $1$.
\end{itemize}
\end{itemize}
This ends the proof that $E^\prime$ gives a value for the bound at
least as great as that for $E$. The lemma is then proved by a simple
induction on the number of uninhibited vertices.

We return to the inductive proof. Using the lemma we will always take
$q^{\prime}=q/2$ giving probability of survival
$\le \alpha^{ 112d^\prime(v)+3  l(q/2) - 2  \log_2 p^{\prime}-t+1}$
$\le \alpha^{ 112d(v)+3  l(q/2) - 2  \log_2 p^{\prime}-t+1}$.

We consider five cases.
\begin{itemize}
\item $q \ge 2/5$: We have $l(q/2)= l(q)-1$ and $p^\prime \ge p/2$ giving
$$
{\mathbb P}r(survival) \le \alpha^{112d(v)+ 3  (l(q)-1)- 2 (\log_2 p-1)-t+1}
=\alpha^{112d(v)+ 3  l(q)- 2 (\log_2 p)-t}
$$
as claimed.

\item $1/5 \le q < 2/5$ and $p<1/2$: The probability that a neighbour
  of $v$ beeps is less than $q$ so that $p_v$ is doubled with
  probability at least $1-q$ and halved in the remaining cases. In all
  cases $l(q/2)=0$.  Hence $P(survival) \le
  \alpha^{112d(v)-2\log_2(p)-t+1}((1-q)\alpha^{- 2 } +q\alpha^{ 2 })$ and our
  claim is that it is at most
  $\alpha^{112d(v)+3\log_2(5q)-2\log_2(p)-t}$. That is the claim is valid
  since $(1-q)\alpha^{- 1 } +q\alpha^{ 3 }\le\alpha^{3\log_2(5q)}$ in
  the range $1/5 \le q < 2/5$.  (It is valid at $q=1/5$ since
  $4\alpha^{-1}+\alpha^3<5$ and at $q=2/5$ since
  $3\alpha^{-1}+2\alpha^3<5\alpha^3$; between these two limits, the
  left hand side is linear and the right hand side
  ($(5q)^{3\log_2\alpha}$) has a negative second derivative so the
  inequality holds there also.)

\item $1/5 \le q < 2/5$ and $p=1/2$: With probability greater than
  $1-q$ no neighbour of $v$ beeps and then $v$ has probability $1/2$
  of taking the current colour; otherwise $p_v$ remains $1/2$. On
  the other hand, if a neighbour does beep, $p_v$ becomes $1/4$.  In
  all cases $l(q/2)=0$.  Thus the probability of survival $\le
  \alpha^{112d(v)+2-t+1}((1-q)/2+q\alpha^2)$ and the claim is that it is at
  most $\alpha^{112d(v)+3\log_2(5q)+2-t}$. That is the claim is valid if
  $(1-q)\alpha/2+ q\alpha^{ 3 }\le\alpha^{3\log_2(5q)}$ a weaker
  condition than in the previous case.

\item $q < 1/5$ and $p<1/2$: The probability that a neighbour of $v$
  beeps is less than $1/5$ so that $p_v$ is doubled with probability
  at least $4/5$ and halved in the remaining cases. In all cases
  $l(q)$ decreases or is unchanged.  Hence ${\mathbb P}r(survival) \le \alpha^{112d(v)+ 3
    l(q) - 2 \log_2(p)-t+1} ((4/5)\alpha^{- 2 } +(1/5)\alpha^{ 2 })$
  and this is less than $\alpha^{112d(v)+ 3 l(q) - 2 \log_2 p-t}$ as claimed,
  again since $4\alpha^{- 1 }+ \alpha^{ 3 }<5 $.

\item $q < 1/5$ and $p=1/2$: With probability greater than $4/5$ no
  neighbour of $v$ beeps and then $v$ has probability $1/2$ of
  taking the current colour; otherwise $p_v$ remains $1/2$. On the
  other hand, if a neighbour does beep, $q$ decreases and $p_v$
  becomes $1/4$.  Hence ${\mathbb P}r(survival) \le (2\alpha^{112d(v)+ 3 l(q/2)- 2
    \log_2(1/2)-t+1}+ \alpha^{ 3 l(q/2) - 2 \log_2(1/4)-t+1})/5$ $\le
  \alpha^{ 3 l(q)- 2 \log_2(1/2)-t+1}(2+\alpha^{ 2 })/5$ which is at
  most $\alpha^{112d(v)+ 3 l(q) - 2 \log_2(1/2)-t}$ as claimed since
  $2+\alpha^{ 2 }<5\alpha^{-1} $.
\end{itemize}

This completes the proof of the theorem.

The complexity of Algorithm \ref{Scott} is described by:
\begin{theorem}\label{coloriagesans}
The number of phases (slots) taken by the colouring algorithm on any graph with
$n$ nodes and maximum degree $\Delta$ is at most $76\log_2 n+112\Delta$ w.h.p.
\end{theorem}
\begin{proof}
Since initially $p_v=1/2$ and $q_v <n/2$ where the graph has $n$
vertices, we conclude that $t_0 < 3\log_2(5n/2)-2\log_2(1/2)<3\log_2
n+6$
so that after $t\ge 112\Delta +
76\log_2 n+6$ phases, any vertex has probability 
 $\alpha^{3\log_2n+6-(76\log_2n+6)}=n^{-73/36}$ of
survival and the probability that any vertex survives is
 at most $n^{-37/36}=o(n^{-1})$.
\end{proof}
\begin{remark}
The number of colours used by the colouring algorithm is at most
 $76\log_2 n+112\Delta$ w.h.p.
\end{remark}

\subsection{A Las Vegas Colouring Algorithm with the
 Knowledge of an Upper Bound of the Maximum 
Degree in $B_{cd}L$}\label{colwku}
This section presents and analyses a Las Vegas colouring algorithm
in the model $B_{cd}L$,
assuming that the vertices are aware of an upper
bound $K$ on the maximum
degree $\Delta$ of the graph. So we aim to compute a $K+1$
colouring.

Each vertex has a counter (initially, its value is $0$) and a set
of colours: $\{0,\cdots,K\}$.
  Each
phase corresponds to three slots. In the first 
slot an uncoloured vertex tries to get a colour by beeping
with a certain probability if the counter belongs to the set of colours.
When a vertex beeps in the second
slot, this means that it succeeds in choosing a colour (the current value
of the counter), so there is no
need to detect collision in this slot.
Vertices which  hear a beep at slot 2 withdraw the corresponding
colour. 
\begin{algorithm}[h]
\label{algorithm::coulouringBcdL}
\caption{A colouring algorithm with the knowledge of an upper bound of the 
maximum degree in $B_{cd}L$.}
\textbf{Var:}\\
\hspace*{1cm}$K:$ \textbf{Global integer constant upper bound on the maximum
degree of $G$;}\\
\hspace*{1cm}$state$ $\in \{Active, Inactive\}$ \textbf{Init} $Active;$ \\
\hspace*{1cm}$Colours$ $=\{0,\cdots,K\}$; \\
\hspace*{1cm}$Colour$ $\in \{0,\cdots,K\}$ \textbf{Init } $0;$ \\
\hspace*{1cm}$counter$ $\in \{0,\cdots,K\}$ \textbf{Init } $0;$\\
\hspace*{1cm}$slot:$ $Integer$;\\
\Repeat{$state=Inactive$}{
  \Switch{slot}{
 \uCase{1}{
  \lIf{$counter\in Colours$}{beep with probability 
       $\frac 1{2\times \mid Colours\mid}$}
}
 \uCase{2}{
  \If{beeped and no internal collision detection}{ $Colour:=counter;$
    $state:=Inactive;$
    beep;\\ }
  \If{beep heard at slot 2}
    {$Colours:=Colours\setminus\{
    counter\}$} $counter:=(counter + 1) \mod K$ }
}
}
\end{algorithm}

\begin{remark}
We can consider  the {\it modified}
colouring algorithm defined in the following way.
By a {\it cycle} we mean $K$ rounds considering the $K$ colours.
Now, every vertex  uses the value of $|Colours|$ at the start
of each  cycle to decide the beeping probability it uses  throughout
this cycle.
\end{remark}

\subsubsection{Analysis of the Algorithm.}
We have the following theorem:
\begin{theorem} 
Let $G$ be a graph  of size $n$, let
$K$ be an upper bound on the maximum degree of $G$. The
Colouring algorithm computes a $K+1$ colouring of $G$ in at most 
$O\left( K(\log n +\log^2K)\right)$
w.h.p.
\end{theorem} 
\begin{proof}
We consider  the Colouring algorithm in which every vertex has the same
upper bound $K$ on the maximum degree. 
We consider both the {\it basic} algorithm
in which $v$ uses the current value of $|Colours|$ to decide its beeping
probability and also the {\it modified} algorithm in which it uses the value
at the start of the current cycle.
We recall that by a {\it cycle} 
we mean $K$ rounds considering the $|Colours|$ colours.

We consider $P_k$ the probability that
vertex $v$ survives uncoloured over $k$ cycles.

In what follows
\begin{itemize}
\item $i$ ranges over $1..k$,\\
\item $c$ ranges over the $C_i$ colours possible for $v$ at the start of cycle $i$,\\
\item $u$ ranges over the neighbours of $v$ still uncoloured at the start of cycle $i$,\\
\item $p_u(i,c)$ is the probability that $u$ beeps at colour $c$ in cycle $i$.
\end{itemize}
First we consider the probability $p$ that $v$ survives uncoloured in a single round
using a colour $c \in colours(v)$ .
\begin{eqnarray}
p & = & {\mathbb P}r\left(v{\rm~does~not~beep~at~colour~}c{\rm~in~cycle~}i\right)\nonumber\\
   & + & {\mathbb P}r\left(v{\rm~does~beep~and~some~neighbour~}u{\rm~also~beeps}\right)\nonumber
\end{eqnarray}   
but ${\mathbb P}r\left(v{\rm~does~beep}\right) \ge 1/2C_i$
and the beeping probabilities of $v$ and its neighbours
are independent giving
\begin{eqnarray}
p & \le  & \left(1-1/2C_i\right) + {\mathbb P}r\left({\rm some~neighbour~beeps}\right)/2C_i\nonumber\\
   & =    & \left(1-1/2C_i\right)\left(1+{\mathbb P}r\left({\rm some~neighbour~beeps}\right)/(2C_i-1)\right)\nonumber\\
   & \le  & \left(1-1/2C_i\right)\left(1+\sum_u p_u(i,c)/(2C_i-1)\right).\nonumber
\end{eqnarray}
After the first round, $p_u(i,c)$ and $C_i$ are random variables
dependent on what has happened so far,
and we consider the tree of all possible executions up to $k$ cycles,
where each tree node has its own value of $p$.
It is easily shown by induction that
$P_k$ is upper bounded by the maximum over all paths in this tree
of the product of the values of $p$ along the path.
We fix a path which gives this maximum and
bound the product for this path.
We have the probability of surviving cycle $i$
 $\le  (exp(-1/2) * \prod_c  (1+\sum_u p_u(i,c)/(2C_i-1)))$
 $\le  exp(-1/2 +  \sum_c \sum_u p_u(i,c)/(2C_i-1))$
and so
$P_k \le  exp(-k/2 + \sum_i \sum_c \sum_u p_u(i,c)/(2C_i-1))$.

We will give an upper bound on $\sum_i \sum_c \sum_u p_u(i,c)/(2C_i-1)$.

We number $v$'s neighbours  in the initial graph from $1$ to $deg(v)$ in decreasing order
of their {\it lifetime}, that is the number of rounds in which they remain uncoloured:\\
Thus as long as $u_j$ is not coloured the degree of $v$ in the residual graph
is at least $j$ and so $|colours(v)|>j$. 

We write $p_u(i,c)$ as $base + \delta$ where $base = 1/2C_i$
and $\delta$ is what has been added as a result of $colours(u)$ being decreased
before colour $c$ and we will bound $\sum_i \sum_u \sum_c base/(2C_i-1)$ and
$\sum_u \sum_i \sum_c \delta/(2C_i-1)$ separately.

Firstly $base$: in cycle $i$, $v$ has $C_i$ colours available and so has less than $C_i$ neighbours;
each neighbour $u$ has $\sum_c base \le 1/2$, giving, for this cycle,
$\sum_u \sum_c base/(2C_i-1)\le 1/6$ so that $\sum_i \sum_u \sum_c base/(2C_i-1) \le k/6$.

Secondly $\delta$:
For the modified algorithm $\delta=0$. In the basic algorithm,
a vertex $u_j$ initially has $K$ colours
available and when (if) this number decreases from $l$ to $l-1$, $p_u(i,c)$ increases
from $1/2l$ to $1/2(l-1)$ and this increase of $1/2l(l-1)$ affects $\delta$ only
for the, at most, $l-1$ colours still to be considered in this cycle so that
$ \sum_c \delta$ for a cycle is at most $\sum_l 1/2l$, the sum being taken over those $l$
for which the number of colours is reduced from $l$.
This gives an upper bound 
on $ \sum_i \sum_c \delta/(2C_i-1)$ of $\log K/2(2j+1)$ since $C_i>j$
and so $\sum_u \sum_i \sum_c \delta/(2C_i-1) < \sum_j \log K/2(2j+1) < \log^2 K/4$.

Hence, by standard arguments, after $k=O(\log n + \log^2 K)$ cycles for the basic
algorithm or $O(\log n)$ cycles for the modified algorithm, $v$ has probability $o(1/n^2)$ of
remaining uncoloured and the graph has probability $o(1/n)$ of having any uncoloured
vertex.

\end{proof}

\section{A Las Vegas  Algorithm for $2$-hop-colouring in 
 $B_{cd}L_{cd}$ without any Knowledge}\label{2col}

To calculate a $2$-hop-colouring of a graph $G$, we need to calculate a 
colouring of
the ``square'' of $G$, that is the graph with the same vertices as $G$
and an edge between any pair $v$ and $w$ of vertices which either are
neighbours in $G$ or have a common neighbour in $G$.  
Algorithm \ref{Scott} (Section \ref{colouring})
is  modified to perform the computation of the
colouring in the square of $G$, i.e., the $2$-hop-colouring of $G$
in $B_{cd}L_{cd}$.

At slot 1, an active vertex beeps with a certain probability. At slot 2,
a vertex having two beeping neighbours beeps. Thus, a candidate vertex,
which beeps without internal collision and which has no neighbours
having detected a peripheral collision, has beeped alone among vertices
at distance at most $2$ and it becomes coloured. 
At slot 3, an 
active  vertex having heard at least one beep beeps.
Finaly after slot 3,
 an active 
vertex  knows whether at least one vertex beeped at distance at most $2$ 
to possibly change its probability 
(as in Algorithm \ref{Scott})
to be candidate.

\begin{algorithm}
\caption{A Las Vegas $2$-hop-colouring algorithm in $B_{cd}L_{cd}$ without any 
knowledge.}
\textbf{Var:}\\
\hspace*{1cm}$state\in\{active,coloured,turned$-$off\}$ \textbf{Init } 
$active$; \\
\hspace*{1cm}$p:$ $real$ \textbf{Init } $1/2$;\\
\hspace*{1cm}$slot:$ $Integer$;\\
\hspace*{1cm}$colour:$ $Integer$ \textbf{Init } $0$ ;\\
\Repeat{$state=turned$-$off$}{
  \Switch{slot}{
    \uCase{1}{
      \If{$state=active$}{
         colour:=colour+1;\\
         $candidate:=true$ with probability $p$ else $false;$\\
        \If{$candidate$}{beep}
      }
    }
    \uCase{2}{
      \If{peripheral collision at slot 1}{beep}
      \If{$candidate$ and ($not$ internal collision at slot 1) 
and (no beep heard at slot 2)}
           {$state:=coloured$}
     }
    \uCase{3}{ 
       \If{beep heard at slot 1}{beep} 
          \If{$state=active$}
            {\If {(not $candidate$) and (no beep heard at slot 1 and at slot 3)}
             {\If{$(p<1/2)$}{$p:=2\times p$}}
             \Else {p:=p/2}}
           }

    \uCase{4}{ 
       \If{$state=active$}{beep} 
      \If{no beep heard at slot 4 and $state=coloured$}{$state:=turned$-$off$}
             }
}
}

\end{algorithm}

In this context, Theorem \ref{coloriagesans} becomes:
\begin{theorem}\label{2hopcolour}
The number of phases taken by the $2$-hop-colouring algorithm on any graph with
$n$ nodes and maximum degree $\Delta$ is at most 
$76\log_2 n+112\Delta^2$ w.h.p. (the number of slots is at most 
$4\times(76\log_2 n+112\Delta^2)$).
\end{theorem}

\begin{remark}
The same transformation can be done the algorithm given in section
\ref{colwku} when we know an upper bound of the maximum degree.
\end{remark}

\section{A Las Vegas Degree 
Computation Algorithm in $B_{cd}L_{cd}$}\label{degreecdcd}
The $2$-hop-colouring algorithm may be viewed as a degree computation algorithm.
We present in this section a Las Vegas Degree 
Computation Algorithm in $B_{cd}L_{cd}$, Algorithm \ref{degree},
 inspired by the $2$-hop-colouring algorithm
given in Section \ref{2col}. The idea is very simple:
each vertex tries to be counted by its neighbours by beeping alone
among vertices at distance at most two. When it is the case it informs
its neighbours which increment their degree, 
it no longer tries to be counted and
listens until the end of the algorithm for counting its neighbours.
Slot 5 allows a vertex to detect the termination of the computation
of its degree.
\begin{algorithm}\label{degree}
\caption{A Las Vegas degree computation algorithm in $B_{cd}L_{cd}$ without any 
knowledge.}
\textbf{Var:}\\
\hspace*{1cm}$state\in\{active,passive,turned$-$off\}$ \textbf{Init } 
$active$; \\
\hspace*{1cm}$p:$ $real$ \textbf{Init } $1/2$;\\
\hspace*{1cm}$slot:$ $Integer$;\\
\hspace*{1cm}$deg:$ $Integer$ \textbf{Init } $0$ ;\\
\Repeat{$state=turned$-$off$}{
  \Switch{slot}{
    \uCase{1}{
      \If{$state=active$}{
         $candidate:=true$ with probability $p$ else $false;$\\
        \If{$candidate$}{beep}
      }
    }
    \uCase{2}{
      \If{peripheral collision at slot 1}{beep}
             }
 \uCase{3}{
      \If{beep heard at slot 1}{beep}
             }

    \uCase{4}{
      \If{$candidate$ and ($not$ internal collision at slot 1) 
and (no beep heard at slot 2)}
           {$beep$; $state:=passive;$}
       \If{beep heard at slot 4}
          {$deg:=deg+1$} 
          \If{$state=active$}
            {\If {(not $candidate$) and (no beep heard at slot 1 and at slot 3)}
             {\If{$(p<1/2)$}{$p:=2\times p$}}
             \Else {p:=p/2}}
           }

    \uCase{5}{ 
       \If{$state=active$}{beep} 
      \If{no beep heard at slot 5 and $state=passive$}{$state:=turned$-$off$}
             }
}
}

\end{algorithm}

\begin{remark}
The degree algorithm allows each vertex to know its degree.
\end{remark}

We deduce from Theorem \ref{2hopcolour} that:

\begin{theorem}
The number of phases taken by the degree algorithm on any graph with
$n$ nodes and maximum degree $\Delta$ is at most $76\log_2 n+112\Delta^2$ w.h.p.
(the number of slots is $5\times (76\log_2 n+112\Delta^2)$).
\end{theorem}
\par\newpage
\section{Emulating $B_{cd}$ or $L_{cd}$ in $BL$}\label{emuler}
This section presents randomised emulation procedure
(Algorithm  \ref{algorithm::emulateBcd}
 and Algorithm \ref{algorithm::procemulLcd}) of $B_{cd}$ and $L_{cd}$ in
$BL$.

The first procedure, $EmulateB_{cd}inBL()$, 
emulates a beeping slot in $B_{cd}$
 in $BL$.
The second, $EmulateL_{cd}inBL()$, emulates a listening slot
in $L_{cd}$  in  $BL$. Both  procedures are Monte Carlo procedures and
parametrized with an integer $k>1$ and an output boolean parameter
$collision$ which indicates whether a collision has been detected. 
The parameter $k$ controls the probability of error for the collision
detection.

Let $v$ be a vertex. Let $k$ be a vertex.
 We denote by $s$ the signature  of
the vertex $v$ which is the word formed by $k$ bits generated uniformly at 
random; it is denoted by $s:=gen(k)$.

Before any emulation each vertex generates its signature $s$ which depends
on $k$: $s:=gen(k);$ then it uses the following procedures (Algorithm 6 and
Algorithm 7).

{\small    
\begin{algorithm}[!h]
    \label{algorithm::emulateBcd}
    \caption{A Procedure to emulate a $B_{cd}$  in  the $BL$ model.}  
    \textbf{Procedure }$EmulateB_{cd}inBL$(\textbf{IN:}{$s:$ word of bits associated to the vertex}\textbf{;} \textbf{OUT: }{$collision: boolean$})
    
    $collision := false;$
    
    $ i := 0;$
    
    \Repeat{$i =  k$ }{
           
      \lIf{$s[i]=0$}{ beep in slot 1; listen in slot 2}
      
      \lElse{listen in slot 1; beep in slot 2}
      
      \lIf{a beep was heard}{$collision := true$}
       $i := i+1$
      
    }
    
    \textbf{End Procedure}
\end{algorithm}

\begin{algorithm}[!h]
  \label{algorithm::procemulLcd}
  \caption{A Procedure to emulate $L_{cd}$  in the $BL$ model.}  
  \textbf{Procedure }$EmulateL_{cd}inBL$(\textbf{IN:}{$k$: Global integer constant}\textbf{;} \textbf{OUT: }{$collision: boolean$} )
    
  $collision := false;$
    
  $i := 0;$
    
  \Repeat{$i =  k$}{          
      
    listen in slot1;
    
    listen in slot2;
      
    \lIf{two beeps were heard}{$collision := true$}
      
    $i := i+1$
  }
  
  \textbf{End Procedure}
\end{algorithm}
}
The value of $k$ depends on the bound of the error probability we require,
a straightforward adaptation of the analysis done in Section  \ref{collision}
gives:
\begin{lemma}\label{emule}
For any $\varepsilon>0$, and any  $n>0$:
\begin{enumerate} 
\item if $k=\lceil \log_2\left(\frac n\varepsilon\right) \rceil$, then, the procedures are correct on $G$ with
probability $1-\varepsilon$, 
\item if $k=\lceil \log_2\left(\frac 1 \varepsilon\right) \rceil$, then, for any vertex $v$, the procedures are correct on $v$ 
with probability $1-\varepsilon$, 
\item if $k = \lceil 2\log_2(n)\rceil$,  then, the procedures are correct on $G$ w.h.p. 
\end{enumerate}
\end{lemma}
\begin{remark}
In the emulation procedures, the for-loops are controlled by $k$ thus
the first item of Lemma \ref{emule} needs  knowledge of $n$ and $\varepsilon$,
the second item needs only  knowledge of $\varepsilon$ and the last
item needs  knowledge of $n$.
\end{remark}

\section{Computing the Degree of each Vertex in $BL$}

This section illustrates emulation procedures by applying
them to the Las Vegas degree algorithm  presented in Section \ref{degreecdcd}
which computes the degree of each vertex in  $B_{cd}L_{cd}$.
We obtain  a Monte Carlo Algorithm, denoted Algorithm 5',
 which  computes the degrees in the $BL$ model. 

We will need two new boolean variables $collision_B$ and
$collision_L$. 

First each vertex generates its signature $s$.
Then we modify Algorithm 5 as follows.

Collisions must be detected only in slot 1 of Algorithm 5. 
As is explained in Remark
\ref{listen}, in this slot, active vertices which do not beep listen.
Thus the instruction in slot 1:

{\bf if} $candidate$ {\bf then}\\
\hphantom{aaaaaa}  beep\\
\noindent
becomes:

{\bf if} $candidate$ {\bf then} $EmulateB_{cd}inBL(s,collision_B)$\\
\hphantom{aaa}{\bf Else} $EmulateL_{cd}inBL(k,collision_L);$

\noindent
and the first  instructions in
slot 2:

  $ic:=$internal collision;\\
     \hphantom{aa} $pc:=$ peripheral collision;\\
    
 become:

 $ic := collision_B$;\\ 
 \hphantom{aa} $pc := collision_L$. 

 The other instructions in the algorithm are not changed.

Finally, Algorithm 5' (a degree computation algorithm in $BL$)
is the concatenation of $gen(k)$ and Algorithm 5 modified as explained
above.
Then, we deduce from Lemma 6.1 the following results:
\begin{theorem}
For any graph $G$ of size $n$ and any $0<\varepsilon<1$:
\begin{itemize}
\item 
if $k=\lceil \log_2\left(\frac n\varepsilon\right) \rceil$,
Algorithm 5' computes the degrees in $G$ in 
$O\left((\log n + \Delta^2)(\log(\frac{n}{\varepsilon}))\right)$, 
and the result is correct with probability at least  $1-\varepsilon$.
\item
 If $k=\lceil \log_2\left(\frac 1 \varepsilon\right) \rceil$,
  each vertex $v$ computes its degree in 
$O\left((\log n + \Delta^2
)(\log(\frac1{\varepsilon}))\right)$, 
and the result is correct with probability at least  $1-\varepsilon$.
\item 
 If $k = \lceil 2\log_2(n) \rceil$,
Algorithm 5' computes the degrees in $G$ in 
$O\left((\log n+\Delta^2)\log n\right)$, and the result is correct 
with probability $1-o\left(\frac 1 n \right)$.
\end{itemize}
\end{theorem}
\begin{remark}
The same transformation can be done for the colouring or the $2$-hop-colouring
algorithms.
\end{remark}

\section{Conclusion}
We present in this paper algorithms which detect collisions in the 
weakest beeping model with a logarithmic complexity. 
Then we consider more powerful beeping models
which enable simple and efficient solutions to the colouring
problem, to the 2-hop-colouring problem and to the degree computation.
Finally, thanks to emulation procedures based on collision detection
we give solutions to these problems in the weakest beeping model 
having a time complexity increased by a logarithmic factor.
\bibliographystyle{alpha} \bibliography{biblio}

\end{document}